\let\accentvec\vec
\documentclass{llncs}

\usepackage[english]{babel}
\usepackage[T1]{fontenc}
\usepackage[utf8]{inputenc}
\usepackage{paralist}
\usepackage{url}
\usepackage{cite}

\usepackage{rotating}
\usepackage{graphicx}
\graphicspath{{./assets/}}
\DeclareGraphicsExtensions{.pdf}

\usepackage{float}
\usepackage[flushleft]{threeparttable}
\usepackage{framed}

\usepackage[colorinlistoftodos,textsize=tiny]{todonotes}

\let\spvec\vec
\let\vec\accentvec
\usepackage{amsmath}

\let\vec\spvec

\usepackage{amssymb} %
\usepackage[only,rrbracket,llbracket]{stmaryrd} %

\usepackage{algpseudocode}

\usepackage{listings}
\usepackage{color}
\definecolor{codegreen}{rgb}{0,0.6,0}
\definecolor{codegray}{rgb}{0.5,0.5,0.5}
\definecolor{codepurple}{rgb}{0.58,0,0.82}
\definecolor{backcolour}{rgb}{0.95,0.95,0.92}
\lstdefinestyle{my_style}{
    numberstyle=\tiny,%
    basicstyle=\footnotesize,
    breakatwhitespace=false,
    breaklines=false,
    keepspaces=true,
    numbers=left,
    numbersep=5pt,
    showspaces=false,
    showstringspaces=false,
    showtabs=false,
    tabsize=4,
    frame=single,
    xleftmargin=0.415cm,
    xrightmargin=0.13cm,
    numberbychapter=false
}

\usepackage{cancel}
\usepackage{eurosym}

\usepackage{multirow}

\usepackage{stackengine}
\usepackage{scalerel}

\usepackage[font=footnotesize,skip=0pt]{caption}

 \newcommand{\specialcellc}[2][c]{%
  \begin{tabular}[#1]{@{}c@{}}#2\end{tabular}}

\let\strong\textbf

\let\origthelstnumber\thelstnumber
\makeatletter
\newcommand*\Suppressnumber{%
  \lst@AddToHook{OnNewLine}{%
    \let\thelstnumber\relax%
     \advance\c@lstnumber-\@ne\relax%
    }%
}

\newcommand*\Reactivatenumber[1]{%
  \lst@AddToHook{OnNewLine}{%
   \let\thelstnumber\origthelstnumber%
   \setcounter{lstnumber}{\numexpr#1-1\relax}%
  }%
}

\AtBeginDocument{}

\newcommand{\range}{\mathrm{range}}
\newcommand{\Id}{\mathrm{Id}}
\newcommand{\kernel}{\mathrm{ker}}

\newcommand{\attack}[2]{\langle #1, #2 \rangle}
\newcommand{\isbydef}{\mathrel{\stackrel{\mbox{\tiny def}}{=}}} 
\newcommand{\know}[1]{K(#1)}

\title{Data Minimisation: a Language-Based Approach (Long Version)}

\author{Thibaud Antignac\inst{1} \and David Sands\inst{1} \and Gerardo Schneider\inst{2}}

\institute{Chalmers University of Technology, 412 96 G\"oteborg, Sweden.\\
\email{\{thibaud.antignac,dave\}@chalmers.se}
\and
Gothenburg University,
412 96 G\"oteborg, Sweden.\\
\email{gerardo@cse.gu.se}}

\begin{document}

\pagestyle{plain}
    
\maketitle

\begin{abstract}
Data minimisation is a privacy-enhancing principle considered as one of the pillars of personal data regulations. 
This principle dictates that personal data collected should be no more than necessary for the specific purpose consented by the user.
In this paper we study data minimisation from a programming language perspective.
We assume that a given program embodies the purpose of data collection, and define a data minimiser as a pre-processor for the input which reduces the amount of information available to the program without compromising its functionality.
In this context we study formal definitions of data minimisation, present different mechanisms and architectures to ensure data minimisation, and provide a procedure %
to synthesise a correct data minimiser for a given program.
\end{abstract}

\section{Introduction} %
\label{sec:introduction}

A standard interpretation of sequential programs is to interpret them as functions taking data as input, performing some computation and giving a result.
The result should somehow be dependent on the input, so provided enough ``good'' data is given as input the program should be able to perform computations and give the expected result in conformance with its functional specification.
There are many reasons for not providing more input data than needed.
Besides time and space, other non-functional requirements can benefit from minimisation of the input.
This is particularly relevant when processing personal, eventually sensitive, data.
According to the Article 5 of the \emph{General Data Protection Regulation} proposal (GDPR) ``Personal data must be [\dots] limited to what is necessary in relation to the purposes for which they are processed''~\cite{gdpr2012}.\footnote{The \emph{General Data Protection Regulation} (EU ---2016/679) was adopted on 27 April 2016, and it will enter into application 25 May 2018. Unlike a Directive it does not require any enabling legislation to be passed by governments.} 
Note that this principle exists in other regulations and is sometimes referred to as ``collection limitation'' when it focuses on the particular step of collecting data.
This is for instance the case of the \emph{Fair Information Practice Principles} (FIPPs)~\cite{fipp1973} in USA, and the \emph{Guidelines on the Protection of Privacy and Transborder Flows of Personal Data}~\cite{oecd2013} proposed by the Organisation for Economic Co-operation and Development (OECD).

If we consider data minimisation from the regulatory point of view, the input data not semantically used in the program should neither be collected nor processed, 
given that the adversary is the data processor\footnote{``Data controllers'' and ``data processors'' are legal roles used to define obligations and liabilities of the parties. We indistinctly use the term ``data processor'' in this paper as we are interested in designating the party that technically processes the data.} 
itself, having thus the possibility to exploit the inputs.\footnote{In other scenarios the adversary only has access to the outputs (cf.~\cite{smith2009}).}
As a consequence the attacker knows all the information available after the input is collected (before the program execution).

Given that \emph{input data = necessary data + extra data}, and since the program should execute equally well without any extra data, we have that  \emph{input data $\geq$ necessary data}.
The goal of the \emph{data minimisation process} is thus to minimise the input data so only what is necessary is given to the program.
Whenever the input data exactly matches what is necessary we may say that the minimisation is the best. Best minimisation is, however, difficult to achieve in general among other things because it is not trivial to exactly determine what is the input needed to compute each possible output.

Let us consider a simple program to exemplify the notion of data minimisation.
The purpose of the program in  Fig.~\ref{lst:benefits_simple_program} is to compute the \emph{benefit level} of employees depending on their salary (that we assume to be between \${$0$} and \${$100000$}).
For the sake of simplicity, in what follows we do not assume any particular distribution over their domain for the inputs. We thus drive our analysis on worst-case assumptions.
A quick analysis shows that the range of the output is $\left\{ \texttt{false}, \texttt{true} \right\}$, and consequently the data processor does not need to precisely know the real salaries of the employees to determine the benefit level. 
In principle each employee should be able to give any number between $0$ and $9999$ as input if they are eligible to the benefits, and any number between $10000$ and $100000$ otherwise, without disclosing their real salaries.

\begin{figure}[t!]
    \centering
\begin{lstlisting}[style=my_style]
input(salary)
benefits := (salary < 10000)
output(benefits)
\end{lstlisting}
    \caption{Running example $P_\textit{bl}$ to compute a benefits level.}
    \label{lst:benefits_simple_program}
    \vspace*{-0.5cm}
\end{figure}

In this paper we address the issue of data minimisation with the following contributions:
\begin{itemize}
	\item We present different mechanisms and architectures to ensure data minimisation (Section~\ref{sec:the_how_and_the_when_of_data_minimisation}),
    \item We provide a formal definition of the concept of \emph{minimisers for a program} used to achieve data minimisation (Section~\ref{sec:data_minimisers}) with respect to an explicit attacker model (Section~\ref{sec:attacker_model}),
    \item We show how we can compare minimisers, leading to the definition of \emph{best minimiser} for a given program (Section~\ref{sec:comparing_pre_processors}), and
    \item We propose a (semi-)procedure to generate minimisers and prove its soundness and termination (Section~\ref{sec:procedure_for_data_minimisation}).
\end{itemize}
In addition, we discuss a proof-of-concept implementation to compute minimisers, based on symbolic execution and the use of a SAT solver, along with two examples (Section \ref{sec:implementation_and_examples}).
Finally we discuss related work (Section~\ref{sec:related_works}) before concluding and giving future research directions (Section~\ref{sec:conclusion}).

\section{The When, How, and Where of Data Minimisation}
\label{sec:the_how_and_the_when_of_data_minimisation}

Having informally introduced what minimality and minimisation are, the natural question arising is how to achieve it. 
Indeed, various mechanisms can be used for data minimisation, as detailed in Section~\ref{sub:minimisation_mechanisms} and can be performed in various architectures, as described in Section~\ref{sub:minimisation_architectures}.

\subsection{Minimisation Mechanisms}
\label{sub:minimisation_mechanisms}

The mechanisms used to test, check, verify, or enforce 
data minimisation appear at different stages of the software lifecycle as illustrated in Table~\ref{tab:enforcing_data_minimisation}.

\begin{sidewaystable}
\centering
\caption{Coverage, input, and output of minimisation mechanisms. }
\label{tab:enforcing_data_minimisation}
\begin{tabular}{|l|c|c|c|c|c|}
\hline
\multicolumn{1}{|c|}{\textit{Timing of mechanism}} & \textit{Kind of analysis}                           & \textit{Type of mechanism}    & \textit{Coverage}                                                             & \textit{Input}                                                                       & \textit{Output}                 \\ \hline \hline
\multirow{6}{*}{before deployment}                 & \multirow{2}{*}{static}                             & verification                  & \multirow{2}{*}{\specialcellc{all possible executions\\(symbolic execution)}} & \multirow{4}{*}{\specialcellc{white box program\\specification}}                     & minimality                      \\ \cline{3-3} \cline{6-6} 
                                                   &                                                     & synthesis                     &                                                                               &                                                                                      & minimiser                       \\ \cline{2-4} \cline{6-6} 
                                                   & \multirow{2}{*}{\specialcellc{static and\\dynamic}} & test                          & \multirow{4}{*}{\specialcellc{some possible executions\\(test strategy)}}     &                                                                                      & minimality                      \\ \cline{3-3} \cline{6-6} 
                                                   &                                                     & refinement                    &                                                                               &                                                                                      & minimiser                       \\ \cline{2-3} \cline{5-6} 
                                                   & \multirow{2}{*}{dynamic}                            & test                          &                                                                               & \multirow{2}{*}{\specialcellc{black box program\\specification}}                     & minimality                      \\ \cline{3-3} \cline{6-6} 
                                                   &                                                     & refinement                    &                                                                               &                                                                                      & minimiser                       \\ \hline
\multirow{6}{*}{online}                            & \multirow{4}{*}{\specialcellc{static and\\dynamic}} & \multirow{2}{*}{verification} & \multirow{4}{*}{\specialcellc{all executions\\(symbolic execution)}}          & \multirow{4}{*}{\specialcellc{white box program\\specification\\current execution}}  & \multirow{2}{*}{minimality}     \\
                                                   &                                                     &                               &                                                                               &                                                                                      &                                 \\ \cline{3-3} \cline{6-6} 
                                                   &                                                     & \multirow{2}{*}{synthesis}    &                                                                               &                                                                                      & \multirow{2}{*}{representative} \\
                                                   &                                                     &                               &                                                                               &                                                                                      &                                 \\ \cline{2-6} 
                                                   & \multirow{2}{*}{dynamic}                            & monitoring                    & \multirow{2}{*}{\specialcellc{some executions\\(previous executions)}}        & \multirow{2}{*}{\specialcellc{previous executions\\current execution}}               & minimality                      \\ \cline{3-3} \cline{6-6} 
                                                   &                                                     & refinement                    &                                                                               &                                                                                      & representative                  \\ \hline
\multirow{5}{*}{offline}                           & \multirow{3}{*}{\specialcellc{static\\and log}}     & \multirow{3}{*}{verification} & \multirow{3}{*}{\specialcellc{all executions\\(symbolic execution)}}          & \multirow{3}{*}{\specialcellc{white box program\\specification\\all execution logs}} & \multirow{5}{*}{minimality}     \\
                                                   &                                                     &                               &                                                                               &                                                                                      &                                 \\
                                                   &                                                     &                               &                                                                               &                                                                                      &                                 \\ \cline{2-5}
                                                   & \multirow{2}{*}{log}                                & \multirow{2}{*}{monitoring}   & \multirow{2}{*}{\specialcellc{some executions\\(previous executions)}}        & \multirow{2}{*}{all execution logs}                                                  &                                 \\
                                                   &                                                     &                               &                                                                               &                                                                                      &                                 \\ \hline
\end{tabular}
\end{sidewaystable}

\subsubsection*{Before deployment}

Three kinds of analysis to perform data minimisation or check minimality can be used 
before deployment: purely static, purely dynamic, or both static and dynamic.
Purely static mechanisms are used to perform a verification of the minimality property or a synthesis of a minimiser, ie., a pre-processor of the input which feeds minimalised data to the program(such as the one presented in Figure~\ref{lst:benefits_simple_minimiser} for the example of $P_\textit{bl}$ given in introduction in Figure~\ref{lst:benefits_simple_program}).
The guarantee provided by such a static analysis is extremely strong as it covers all the possible executions.
To achieve this the source code of the program is needed along with its specification.
However, static analysis can be expensive and may not alway be applicable.
In that case one may turn to more lightweight techniques.

At the other end of the spectrum, it is possible to use purely dynamic mechanisms to perform tests of the minimality property or to refine a pre-processor to make it more minimising.
In this case, the test strategy will choose some possible inputs and test whether or not some other inputs invalidate the minimality property.
For example, we could test $P_\textit{bl}$ with some input values,
such as $\texttt{0}$, $\texttt{1000}$, and $\texttt{10000}$; the input values $\texttt{0}$ and $\texttt{1000}$ lead to the same output and would thus reveal that the program has not been fully minimised.
However, an absence of common ouputs does not prove the program has been minimised. 
Giving weaker guarantees than purely static techniques, purely dynamic techniques only need a black box executable program to run tests (along with its specification).

A mix between static and dynamic analysis can be used by performing a static analysis on a set of dynamically selected inputs (following a test strategy).
For example, it would be possible to perform a formal verification only for input value $\texttt{0}$ for $P_\textit{bl}$.
In this case, a result as good as the one obtained with the purely static approach could be obtained for this particular input, returning the set of values $\{0, \dots, 9999\}$ (all giving the same output).
However, it does not give any guarantee for the other input values.
This approach gives strong guarantees on the inputs covered, but does not cover the full domain of program inputs (contrary to purely static approaches). 

After deployment, both online and offline mechanisms can be run.

\subsubsection*{Online}

Concerning online mechanisms, only a mix of static and dynamic analysis or a purely dynamic approach can be used.
Indeed, at this stage, it becomes meaningless to perform a purely static analysis as such an analysis is needed only once (it covers all the possible executions and thus is better run before deployment as explained previously).
A mixed analysis is based on the current execution.
In this case, it is still possible to reach strong guarantees about minimality or to get an input representative, ie., another input value hiding the the actual input value while still preserving correctness, which ensures minimality.
Indeed, while the purely static approach consisted in generating a minimiser returning a representative for all the possible inputs, only the representative corresponding to the current execution is computed here.
Returning to the example of $P_\textit{bl}$, a client whose salary is of $\texttt{8000}$ will execute the minimisation mechanism before disclosing this data.
Thus, a representative (for instance, $\texttt{0}$) will be fixed at this time to hide the actual salary to the server without changing the result of the computation.
As a consequence, the computation is more efficient but still relies on a white box program along with its specification and data relative to the current execution.

If such conditions are not acceptable, it is also possible to get better efficiency at the cost of weaker guarantees by using a purely dynamic approach.
This consists in testing whether or not the previous executions with similar inputs led to the same output, meaning the minimality property is not reached.
Returning again to our example, let us say the salary is to be disclosed every month and there has already been a disclosure of $\texttt{7000}$ the previous month.
Instead of disclosing its current salary of $\texttt{8000}$, the client will reuse the previously disclosed value of $\texttt{7000}$ as it leads to the same output value.
In this case, it has to be highlighted that even if the test or the refinement is performed against all the previous executions, the minimality property cannot be guaranteed as there might be more information leaking without it being witnessed by the past behaviours.

Instead of relying on previous executions, the dynamic analysis could also rely on arbitrary values, following a test strategy, as depicted for mechanisms run before deployment.
In this case, an access to the black box program is necessary instead of the previous executions.
For the sake of simplicity, we do not show this case in the Table~\ref{tab:enforcing_data_minimisation} as it does not appear to make an interesting case in practice.

\subsubsection*{Offline}

Finally, offline mechanisms rely on execution logs to audit minimality. 
They also allow to get guarantees as strong as previously in the mixed static and log analysis.
As for the previous kinds of analysis, this requires a white box program along with its specification.
In this case, all the behaviours belonging to the execution log are verified individually in a way very similar to the online mechanism.

A weaker analysis, similar to the online dynamic one, can also be performed here, only relying on logs.
This monitoring approach only consider all the executions recorded in the logs and are only able to find witnesses of non-minimality.
For instance, if two executions have been logged disclosed inputs $\texttt{7000}$ and $\texttt{8000}$ (both leading to the same output), the auditor knows that more data than needed has been disclosed.

It could also be possible here to monitor each past execution against arbitrary values given by a monitoring strategy (similarly to test strategies).
This would require an access to the black box program and would give similarly weak results.
This possibility has not been added to the Table~\ref{tab:enforcing_data_minimisation} as it does not seem useful in practice.

Data minimisation itself can only be performed before deployment or online.
Indeed, after execution, an offline approach is only able to detect breaches in minimality.
As such, it is not a proper privacy-by-design practice and should not be used alone.
However, this approach is still useful as it is the more likely to be used in already existing systems (as it would be extremely ambitious to hope all systems managing personal data to be redesigned according to privacy-by-design guidelines).
Another use of online and offline mechanisms is their application by an external entity, such as a data protection authority, to perform sporadic controls (with the caveat that such controls could be subject to attacks depending on prior mechanisms run on data by the server).

The way that these mechanisms are applied may depend on the system architecture.
Thus they may provide different guarantees as discussed below.

\subsection{Minimisation Architectures}
\label{sub:minimisation_architectures}

The kind of minimality and minimisation achievable by these mechanisms depend on when and where they are performed as shown in Table \ref{tab:achievement_data_minimisation}.

\begin{table*}[htbp]
\centering
\caption{Kind of minimality applicable and privacy property targeted by minimisation mechanisms in a distributed architecture. }
\label{tab:achievement_data_minimisation}
\begin{tabular}{|l|c|c|c|}
\hline
\multicolumn{1}{|c|}{\specialcellc{\textit{Timing of}\\ \textit{mechanism}}} & \specialcellc{\textit{Location of}\\ \textit{mechanism}} & \specialcellc{\textit{Kind of minimality}\\ \textit{applicable}} & \specialcellc{\textit{Privacy property}\\ \textit{targeted}}       \\ \hline \hline
\multirow{3}{*}{before deployment}                 & clients                        & distributed                            & \multirow{2}{*}{collection minimisation} \\ \cline{2-3}
                                                   & 3rd party                      & \multirow{2}{*}{monolithic}            &                                          \\ \cline{2-2} \cline{4-4} 
                                                   & server                         &                                        & use minimisation                         \\ \hline
\multirow{3}{*}{online}                            & clients                        & distributed                            & \multirow{2}{*}{collection minimisation} \\ \cline{2-3}
                                                   & 3rd party                      & \multirow{2}{*}{monolithic}            &                                          \\ \cline{2-2} \cline{4-4} 
                                                   & server                         &                                        & use minimisation                         \\ \hline
\multirow{3}{*}{offline}                           & clients                        & distributed                            & \multirow{3}{*}{breach detection}        \\ \cline{2-3}
                                                   & 3rd party                      & \multirow{2}{*}{monolithic}            &                                          \\ \cline{2-2}
                                                   & server                         &                                        &                                          \\ \hline
\end{tabular}
\end{table*}

For mechanisms run before deployment or online, different kinds of minimality are applicable and different privacy properties can be targeted depending on the architecture.

\subsubsection*{Kind of Minimality}

Depending on the architecture, different kinds of minimalities can be achieved, denoted {\em monolithic} minimality and {\em distributed} minimality.
Monolithic minimality is a minimality property expressed in the context where knowledge about all input values is accessible at a centralised location where the mechanism can be applied.
On the other hand, distributed minimality is used for the cases where the input values and the associated mechanism are distributed, in which case the mechanism has only the knowledge available at its location.
If the mechanism is performed by a client then only monolithic minimality can be achieved.
On the other hand, if the mechanism is performed by a third-party or by the server, then monolithic minimality is achieved as there is full knowledge about input values.

\subsubsection*{Privacy Properties}

The kind of privacy property that can be achieved depends on when in the data lifecycle the methods are deployed. 
{\em Collection minimisation} prevents data from being collected by the server (this is the case when the client itself, or a trusted third-party, performs the minimisation).
On the other hand, {\em use minimisation} is only able to limit the use of the data (this is the case when the server itself is also in charge of performing the minimisation\footnote{One could argue that it is already too late in this case, but the server could belong to a multi-server organisation, thus minimising potential subsequent disclosures (not treated in this paper).}).

We note that the only way to implement 
collection minimisation while ensuring monolithic minimality in a distributed setting is by using a third party to perform the minimisation.
Indeed, it is not possible to reach a property stronger than distributed minimality when the minimisation is applied client-side. 

Figure~\ref{fig:architecture} gives an illustration of monolithic vs. distributed minimisation in the case of a static minimiser generated before deployment, and run by the client (thus guaranteeing collection minimisation). %

\begin{figure}[t!]
    \centering
     \includegraphics[width=1\linewidth]{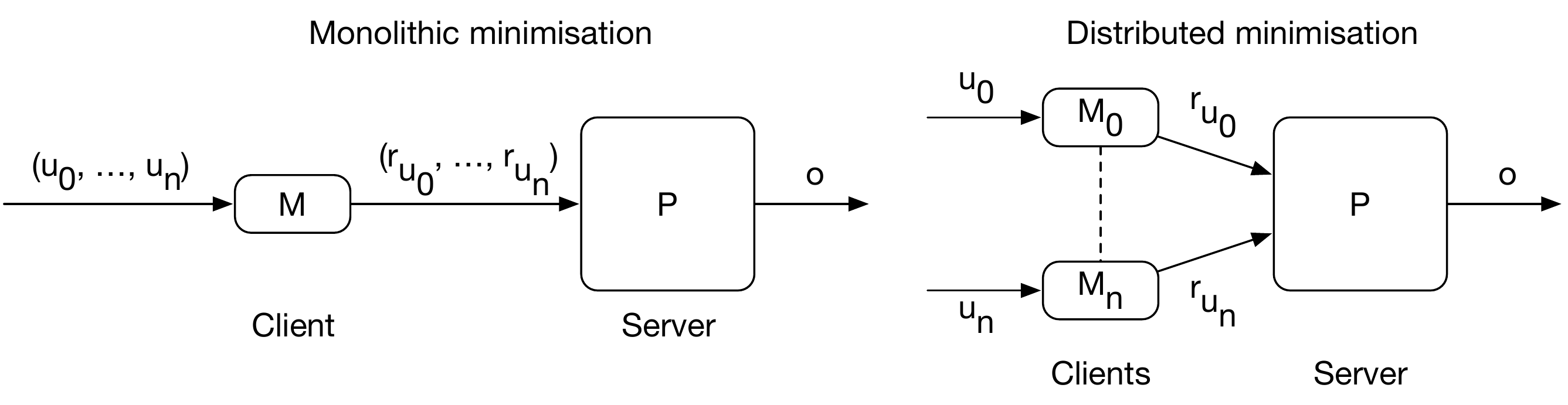}
    \caption{Monolithic and distributed minimisation architectures.}
    \label{fig:architecture}
    \vspace*{-0.5cm}
\end{figure}

In the following, we detail the attacker model used to formally characterise the notion of minimality.
In the following, unless otherwise stated, will focus on the static mechanism built before deployment for verification purposes; this is the most complex case, and 
 the other mechanisms can be seen as weakened versions of this one.

\section{Attacker Model}
\label{sec:attacker_model}

In order to define data minimisation we have to 
consider an
explicit attacker model. For example, if the computation performed on
the data is really just the intended program $P$, then why worry about
minimisation? In practice a malicious%
\footnote{In security literature, a malicious attacker is an attacker which does not follow the protocol while an attacker only attempting at extracting more information from what he got is called honest-but-curious. In the privacy context, the (security) honest-but-curious attacker cannot be considered as honest because personal data should not be used for another purpose than the one specified and consented for by the data subject. As a consequence, this is also considered as a malicious behaviour.} server may have a \emph{secondary use} for
the data (defined as ``using personal data in a different way than the primary reason why they were collected''~\cite{bussard2011}).
We define here an attacker model having the following components: 
(i) an explicit model of a hidden secondary use of the data, and 
(ii) an explicit model of what the attacker learns about the data from observing the result of processing the data. 
We will develop a definition of 
minimisation which guarantees that an attacker learns no more about the
private data than he can deduce from the legitimate output of the
program.  

To model the hidden secondary use, we suppose the attacker uses a second program denoted by a function
$h \in \mathcal{I} \to \mathcal{O'}$. Thus the attacker is defined by
a pair of functions $p$ and $h$ (the legitimate use and the hidden use of the
data, respectively), and the attacker's computation is given by the
function $\attack{p}{h} \in \mathcal{I} \to
\mathcal{O} \times \mathcal{O'}$ defined by
 \[
\attack{p}{h}(i) \isbydef (p(i),h(i))  %
 \]

The goal of the attacker is to learn something about the inputs by observing the
output of $\attack{p}{h}$. In the following section we will show that if the input is first
minimised using a best possible minimiser, then what the attacker learns from $\attack{p}{h}$ is
no more than what is learned by ${p}$ alone.
To do this we
model the attacker knowledge explicitly, as has become more common in
recent formulations of the semantics of information flow (e.g. \cite{Askarov:Sabelfeld:Gradual,balliu2011epistemic,askarov2012}).

\begin{definition}
    [Knowledge]\label{def:knowledge}
Let $u \in \mathcal{I}$ be an input, and $f$ a function in $\mathcal{I} \to
\mathcal{O}$. 
    We say the \emph{knowledge} about $u$ obtained through $f$ by
    observing $f(u)$, written $\know{f,u}$, is defined to be the set
    $\{ v ~ | ~ f(u) = f(v), v \in \mathcal{I} \}$.
\end{definition}

Thus $\know{f,u}$ is the largest set of possible input values that
could have led to the output value $f(u)$. This corresponds to an
attacker who knows the function $f$ and the output $f(u)$ and makes
perfect logical deduction about the possibles values of the input. 
Note that the smaller the knowledge
set, the less uncertainty the observer has over the possible value of
the input. 

The knowledge relation induces an order on programs: the bigger the knowledge set for a given input, the lesser a program will disclose information about the input. This is reasonable given that if more input values are mapped to the same output, more difficult it will be to guess what is the exact input when observing the output of a computation. We have then the following order between programs:

\begin{definition}[Disclosure ordering]\label{def:disclosure}
    We say $f$ {\em discloses less or as much information as} $g$, written $f \sqsubseteq g$, iff for any $u \in \mathcal{I}$ we have $\know{g,u} \subseteq \know{f,u}$.
    \label{def:exkmin}
\end{definition}

If $f$ and $g$ are knowledge equivalent, we write $f \equiv g$. 
For example, functions denoted by $(\cdot\, \mathrm{mod} 2)$ and $(\cdot\, \mathrm{mod} 4)$ (with postfix notation) are related by $(\cdot\, \mathrm{mod} 2) \sqsubseteq (\cdot\, \mathrm{mod} 4)$ (knowing the last bit discloses less information than knowing the last two bits).
However, $(\cdot\, > 0)$ is incomparable to either of them.

\section{Data Minimisers}
\label{sec:data_minimisers}

In this section we develop a definition of a \emph{minimiser} for a program. %
We will consider two types of minimiser depending on how the data is
collected: \emph{monolithic} and \emph{distributed}.
Monolithic collection describes the simple case when the personal data is collected from a single location (e.g. a single data subject).
On the other hand, distributed collection occurs when personal data is collected from different sources.
In the latter case we assume these sources do not interact to achieve a multi-party data minimisation.
This assumption is important as it has a strong impact to the kind of minimisers which may be feasible in practice.

The execution of a program consists in the computation of an output, depending on the inputs provided.
The domain of possible inputs can be arbitrarily large and complex.
Running a data minimiser before a program is intended to capture the idea that the
program will be given less data.
In the best case, only the data actually needed to produce the result will be provided. %
Thus it implies that
the system does not receive more sensitive data than it
needs. 

We will address the whole data minimisation process by the following steps: (i) data minimiser generation (from the program, performed client-, server-, or third-party-side), (ii) data minimisation \emph{per se} (ie. data minimiser execution, performed client-side), (iii) data collection (ie. representative transmission to the server, performed client-side), and (iv) program execution (performed server-side). Note that a minimiser only needs to be generated only once for a given program: the same can be used for all clients. This section and the following one are mainly concerned with the first of these steps (the other steps being simple executions and communications). The improvement in the minimality of a given program is achieved by constructing a \emph{minimiser} $M$. %

In the following we assume the data processor to be the malicious attacker depicted in Section~\ref{sec:attacker_model} and the data sources to be as trustable as they would be without data minimisation.
The data minimiser could be sent as part of the protocol for data collection. 
This may be done by a third-party or by the data processor itself.
Thus we assume that the data source knows how much data is minimised on its side, but may not be able to check the optimality of the whole minimiser.
A third-party provider for the data minimiser could help to get this assurance.
Note, however, that the fact that the specification is being determined by the data processor itself does not weaken (nor strengthen) the security and trust assumptions of the interactions between the parties.

\subsection{Minimisation Semantics}
\label{sub:minimisaton_semantics}

\begin{figure}[t!]
    \centering
\begin{lstlisting}[style=my_style]
input(x1, x2, x3);
y := 0;
if (x2 == x2;)
    then (y := (x3 + x1 - x3);)
output(y);
\end{lstlisting}
    \caption{Example of a semantically unnecessary syntactic necessity.}
    \label{lst:syntactic_necessity}
    \vspace*{-0.5cm}
\end{figure}

Let us consider the program in Fig.~\ref{lst:syntactic_necessity} to discuss the need for a semantic definition of minimisation.
In this program, $\texttt{x2}$ is syntactically needed to evaluate the condition (l.3).
However, this condition always evaluates to $\texttt{true}$.
In the same way, $\texttt{x3}$ is not semantically needed to compute the value of $\texttt{y}$ since it is both added and substracted (l.4).
As a consequence, it would be possible to get the same result by rewriting the program with only the input $\texttt{x1}$ without modifying its behavior.
If $\texttt{x2}$ and $\texttt{x3}$ are personal data, then the semantic approach is better likely to help us to limit the collection than the syntactic approach.
So, the program could be refactored by taking only the input $\texttt{x1}$ while retaining the same output behaviour.
Though this would work, this approach requires a change in both the data collected and in the interface of the program.
 
Instead, we propose to keep the program unchanged and we rely on the idea that the information behind $\texttt{x2}$ and $\texttt{x3}$ (in this specific example) can be protected by providing \emph{fixed arbitrary} values for them, instead of refactoring the program.
This means the program does not need to be modified for the data processor to propose better guarantees.
This approach allows a better modularity and a better integration in the software development cycle.
To see this, let us consider the program shown in Fig.~\ref{lst:benefits_simple_program}.
In this case, any employee earning less than {\$}$10000$ can disclose any figure between $0$ and $9999$ and any employee earning at least {\$}$10000$ can disclose any figure between $10000$ and $100000$ without affecting the output of the program.
Thus a corresponding \emph{data minimiser} could be as shown in Fig.~\ref{lst:benefits_simple_minimiser}, where the representative values are taken to be $\texttt{0}$ and $\texttt{10000}.$ 

\begin{figure}[t!]
    \centering
\begin{lstlisting}[style=my_style]
input(salary)
repr_salary := 0
if (10000 <= salary)
    then (repr_salary := 10000)
output(repr_salary)
\end{lstlisting}
    \caption{Minimiser for $P_\textit{bl}$ (see Figure~\ref{lst:benefits_simple_program}).}
    \label{lst:benefits_simple_minimiser}
    \vspace*{-0.5cm}
\end{figure}

The value output by this minimiser, which is the value of $\texttt{repr\_salary}$, can then be provided as input to the corresponding program (standing for the variable $\texttt{salary}$).
The behaviour of the program will remain the same while the actual salary of the employee is not needlessly passed to the program.
This holds even when an employee earns exactly the amount of $\texttt{repr\_salary}$ since the data processor receiving this value cannot distinguish when it is the case or not.

Following the approach about data minimality, we study in turn both monolithic and distributed cases as shown in Fig.~\ref{fig:architecture} where $M$ and $\bigotimes_{i=0}^n M_i$ are a monolithic and a distributed minimiser respectively (each $M_i$ being a local minimiser), $P$ is the program, $u_i$ are input values, $r_{u_i}$ are input representatives, and $o$ are outputs. We then discuss about security properties achieved through composition of minimisers and programs.

\subsection{Monolithic Case}
\label{sub:minimiser_monolithic_case}

Let us assume a program $P$ constituting a legitimate processing of sensitive (personal) data for which the data subject had consented to the purpose.
We abstract it through a function $p$ which ranges over $\mathcal{I} \to \mathcal{O}$. 
Since we aim at building {\em pre-processors} for the program, we consider a minimiser $m$ with type $\mathcal{I} \to \mathcal{I}$.

As stated previously, monolithic collection refers to the case when there is only one data source from which data is collected before being processed.
This data source can thus minimise the data disclosed while having full knowledge about the input that are provided to the processor.
Thus the output of a data minimiser $M$ can be directly input to the program $P$.

\begin{definition}
    [Monolithic minimiser]
    We say $m$ is a \emph{monolithic minimiser} for $p$ iff (i) $p\, \circ\, m = p$ and (ii) $m\, \circ\, m = m$.
    \label{def:pre-processor_mono}
\end{definition}

Condition (i) on the above definition is basically correctness: it states that using the pre-processor does not change the result of the computation. Condition (ii) ensures that $m$ chooses representative inputs in a canonical way. Put it another way, $m$ has the job of removing information from the input.
This idempotency condition states that $m$ removes all the information in one go.

\subsection{Distributed Case}
\label{sub:minimiser_distributed_case}

In general a computation over private data may require data collected from several independent sources.
We will thus extend our definition to the distributed setting.
This is the general case where a program $\textit{dp}$ is a function of a product of input domains such that $\textit{dp} \in \prod_{i=0}^n \mathcal{I}_i \to \mathcal{O}$. The idea for the distributed minimiser will be to have a \emph{local minimiser} $m_i \in \mathcal{I}_i \to \mathcal{I}_i$ for each source $\mathcal{I}_i$, combined into an input processor as $\textit{dm} = \bigotimes_{i=0}^n m_i$, where for $f \in A \to A'$ and $g \in B \to B'$ we have $f \otimes g \in A \times B \to A' \times B'$.
This is based on the assumption that each argument of $\textit{dp}$ is provided by a different data source and hence requires its own data processing.

\begin{definition}
    [Distributed minimiser]
    We say $\textit{dm}$ is a \emph{distributed minimiser} for $\textit{dp}$ iff (i) $\textit{dm}$ is a monolithic minimiser for $\textit{dp}$ and (ii) $\textit{dm}$ has the form $\bigotimes_{i=0}^n \textit{dm}_i$ (with $\textit{dm}_i$ being $\mathcal{I}_i \to \mathcal{I}_i$ functions (called local minimisers, for $i \in \{0, \dots, n \}$)).
    \label{def:pre-processor_dist}
\end{definition}

The first condition ensures that $\textit{dm}$ is actually a minimiser while the second condition (ii) ensures that each input is treated independently.
This second condition is the key differentiator between monolithic and distributed minimisers. Each function $\textit{dm}_i$ is only provided with an input from $\mathcal{I}_i$. Intuitively, this means that if two input values at a given position belong to the same equivalence class, then it has to be the case that this holds for all values possibly provided for the other positions. For example, if the values $0$ and $1$ from $\mathcal{I}_0$ belong to the same equivalence class, then each pair of input tuples $\langle 0, x, y \rangle$ and $\langle 1, x, y \rangle$ for any $(x, y) \in \prod_{i=1}^2 \mathcal{I}_i$ have to belong to the same equivalence class (two by two, meaning these can be different classes for each pair: $\langle 0, 1, 2 \rangle$ and $\langle 1, 1, 3 \rangle$ may not belong to the same equivalence class, though $\langle 0, 1, 2 \rangle$ and $\langle 1, 1, 2 \rangle$ have to). %
This is formally stated in Proposition \ref{prop:proper_distribution} below, which relies on the definition of the \emph{kernel} of a function. 
\begin{definition}
    [Kernel of a function]\label{def:kernel}
If $f \in \mathcal{I} \rightarrow \mathcal{O}$ then the {\em kernel of} $f$ is defined as 
$\ker(f) = \{ (u,v) \mid f(u) = f(v) \}$.
\end{definition}
So, the kernel $\ker(f)$ induces a partition of the input set where every element in the same equivalence class is mapped to the same output.

\begin{proposition}%
\label{prop:proper_distribution}
    If $m$ is a monolithic minimiser for $p$, then $m$ is a distributed minimiser for $p$ iff 
    for all $(\vec u, \vec v) \in \kernel (m)$, for all input positions $i$, and all input vectors $\vec w$, \[
(\vec w[i \mapsto u_i], \vec w [i \mapsto u_i]) \in \kernel (m)
\]
where the notation $\vec w[i \mapsto u_i]$ denotes a vector like $\vec w$ except at position $i$ where it has value $u_i$. 
\end{proposition}

This proposition gives a data-based characterisation of \emph{data minimisation}.
This will be useful when building minimisers in Section~\ref{sec:procedure_for_data_minimisation}.
Before building minimisers, we explain how to compare them in the next section.

\section{Best Minimisers}
\label{sec:comparing_pre_processors}

Now that we defined minimisers as pre-processors modifying the input, we see that there may exist many different minimisers for a given program. Thus we are interested in being able to compare these minimisers. Indeed, since the identity function is a minimiser for any program $p$, then it is clear that the simple existence of a minimiser does not guarantee any kind of minimality. 
One way to compare minimisers is to compare the size of their ranges -- a smaller range indicates a greater degree of minimality (cf.~Proposition \ref{prop:kpost} below).
But a more precise way to compare them is by understanding them in terms of the \emph{lattice of equivalence relations} \cite{oystein1942}, which has been used to model dependency properties of programs in several works, e.g. \cite{Hunt:PhD,Landauer:Redmond:CSFW93,sabelfeld2001}. 

The set of equivalence relations on $\mathcal{I}$ forms a complete lattice, with the ordering relation given by set-inclusion of their defining sets of pairs. 
The identity relation (which we denote $\Id_\mathcal{I}$) is the bottom element, and the total relation (which we denote $\mathrm{All}_\mathcal{I}$) is the top.

The following proposition provides some properties about the order relation between programs (cf.~Definition \ref{def:disclosure}), including its relation with the kernel.

\begin{proposition}
	[Disclosure ordering properties]
    \label{prop:kpost}
    \begin{compactenum}
        \item\label{prop:one} $f \sqsubseteq g$ iff $\ker(g) \subseteq \ker(f)$
        \item $f \circ g \sqsubseteq g$
        \item $ \attack{f}{f} \sqsubseteq f $
        \item $f \sqsubseteq \attack{f}{g}$
        \item $f \sqsubseteq g$ implies $|\range(\textit{f})| \leq |\range(\textit{g})|$
    \end{compactenum}
\end{proposition}
where the notation $\attack{f}{g}$ a function denoting a mapping of an argument to each one of $f$ and $g$ as defined in Section~\ref{sec:attacker_model}.

\subsection{Monolithic Case} %
\label{sub:monolithic}

In order to reason about minimisers we need to be able to compare them. 
The ordering relation between programs defined above may be specialised to minimisers: $\textit{m}$ is a \emph{better minimiser} for $\textit{p}$ than $\textit{n}$ is iff (i) $\textit{m}$ and $\textit{n}$ are minimisers for $\textit{p}$ and (ii) $\textit{m} \sqsubseteq \textit{n}$.
Of particular interest is the minimiser that fulfill its purpose in the {\it best} way:

\begin{definition}
    [Best monolithic minimiser]
    We say that $m$ is a \emph{best monolithic minimiser} for $p$ iff (i) $m$ is a monolithic minimiser for $p$ and (ii) $m \sqsubseteq n$ for any $n$ a minimiser for $p$.
    \label{def:best_minimiser_mono}
\end{definition}

A minimiser carrying less information than the best minimiser could not convey all the information necessary for the program $p$ to keep the same behaviour under all inputs.
In this simple (monolithic) form, minimisation is easily understood as injectivity:
\begin{proposition}
    [Monolithic best minimiser injectivity]
    A monolithic minimiser $m$ for a program $p$ is a best one iff $p |_{\range(m)}$ is injective (with $p |_{\range(m)}$ the restriction of the program $p$ over the range of $m$).
\end{proposition}

Now we can show that using a minimiser $m$ at the client guarantees that the
attacker $\attack{p}{h}$ learns no more about the input than which can
be learned by observing the output of the legitimate use $p$ (the proof follows from the ordering between minimisers and Proposition~\ref{prop:kpost}):
\begin{theorem}\label{thm:k}
  If $m$ is a best monolithic minimiser for $p$ then %
  for any hidden
  use $h$ we have $\attack{p}{h}\circ m \equiv p$.
  
  \begin{proof}
\begin{align*}
    \begin{array}{rll}
    p~\equiv & p \circ m & (\text{$m$ is a minimiser for $p$})
\\ \sqsubseteq & \attack{p\circ m}{ h \circ m}  &(\text{Prop. \ref{prop:kpost}(4)})
      \\ \equiv & \attack{p}{h} \circ m        &(=)
      \\
\\ \attack{p}{h} \circ m \sqsubseteq & m & (\text{$m$ is a
                                 minimiser for $p$,}\\
                                 && \text{Prop. \ref{prop:kpost}(2)})
\\  \sqsubseteq & \attack{p}{m} & (\text{Prop. \ref{prop:kpost}(4)})
\\\sqsubseteq & \attack{p}{p} & (\text{Def. \ref{def:exkmin}})
      \\ \sqsubseteq & p & (\text{Prop. \ref{prop:kpost}(3)})
  \end{array}
\end{align*}
\end{proof}

\end{theorem}
Finally, we note that:
\begin{theorem}\label{th:monolithic_minimiser_existence}
    For every  program $\textit{p}$ there exists a best monolithic minimiser.
\end{theorem}

\subsection{Distributed Case} %
\label{sub:distributed}

As for the monolithic case, we have an ordering between distributed minimisers. 
The notion of a best minimiser also exists in the distributed case and restricts the monolithic case definition as follows:

\begin{definition}
    [Best distributed minimiser]
    We say $\textit{dm}$ is a \emph{best distributed minimiser} for $\textit{dp}$ iff (i) $\textit{dm}$ is a distributed minimiser for $\textit{dp}$ and (ii) $\textit{dm} \sqsubseteq \textit{dn}$ for any $\textit{dn}$ a distributed minimiser for $\textit{dp}$.
    \label{def:best_minimiser_dist}
\end{definition}

The following result shows that there always exists a best distributed minimiser (see proof in Appendix A).
\begin{theorem}
    \label{th:distributed_minimiser_existence}
    For every distributed program $\textit{dp}$ there exists a best distributed minimiser.
\end{theorem}
Note, however, that there is a price to be paid for having distributed
minimisers: the best distributed minimiser may reveal more
information than the best monolithic minimiser for a given function. 
An example is the Boolean \emph{OR} function where the identity
function is a best distributed minimiser, but in the monolithic case we can minimise further by mapping
$(0,0)$ to itself, and mapping all other inputs to $(1,1)$.

Similarly to what we did in the previous section with Proposition~\ref{prop:proper_distribution}, we give a data-based characterisation of best minimisers as follows.

\begin{proposition}[Data-based best distributed minimisation] \label{prop:best_distr_minimisation}
    If $\textit{dm}$ is a best distributed minimiser for $\textit{dp}$, then 
    for all input positions $i$, for all $v_1$ and $v_2 \in \mathcal{I}_i$ such that $v_1 \neq v_2$, there is some $\vec u \in \range (\textit{dm})$ such that \[
\textit{dp} \left( \vec u[i \mapsto v_1] \right) \neq \textit{dp} \left(  \vec u [i \mapsto v_2] \right).
\]
\end{proposition}

This property will prove useful in the next section to build minimisers.

\section{Building Minimisers} %
\label{sec:procedure_for_data_minimisation}

We describe here a procedure to build data
minimisers. This (semi-) procedure is not complete for the best minimisers 
since obtaining a best minimiser is not computable in general. 
Besides, and more pragmatically, our procedure is built on top of a theorem prover and a symbolic execution engine, and it is thus also limited by the power of such tools in addition to the theoretical algorithmic limits.

Only the main program (and optionally a specification overriding the
implicit specification embodied in the program itself) is needed to generate a minimiser. Data minimisers thus generated are intended to be deployed either monolithically or in a distributed manner with a part residing at each data collection point;
each data source locally executes the local minimising program before forwarding the representative inputs to the data processor (i.e. the program).

Let us consider again the example described in Fig.~\ref{lst:benefits_simple_program} to give an intuition of the way minimisers can be generated.
In order to  generate a data best minimiser for the program we need to reason about  the dependencies between the possible outputs and the inputs considering all possible execution paths.
A representation of the possible paths for this program is shown in Fig.~\ref{fig:benefits_simple_tree} where the digits $1$ and $2$ correspond to the lines of the program, $\mathrm{Sto}$ is the corresponding store, and $\mathrm{PC}$ is the condition to be satisfied at this point (l.3 of the program does not have any effect in this representation).

\begin{figure}[!t]
    \centering
    \includegraphics[width=1\linewidth]{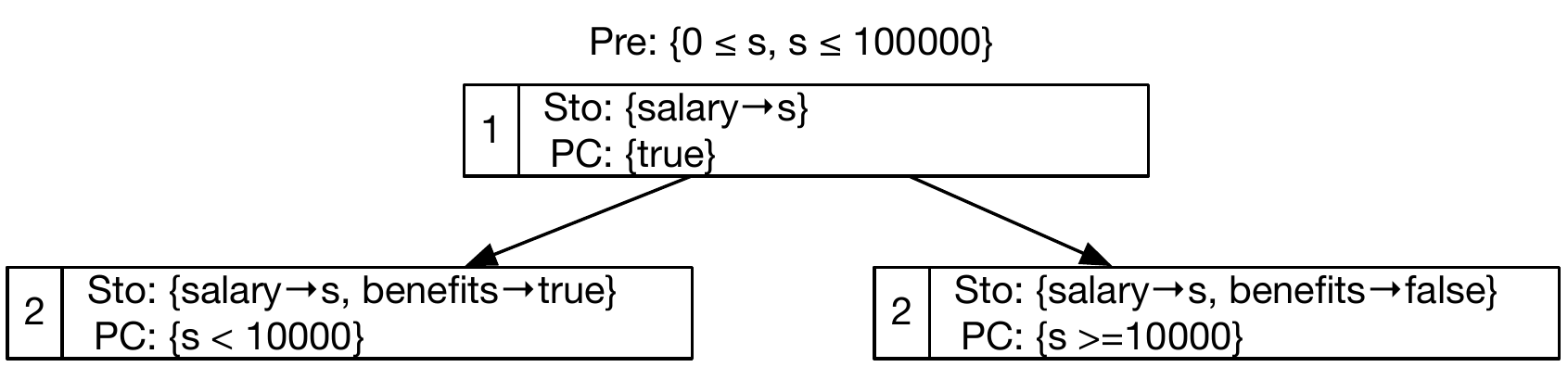}
    \caption{Symbolic execution tree of $P_\textit{re}$ (see Figure~\ref{lst:benefits_simple_program}).}
    \label{fig:benefits_simple_tree}
    \vspace*{-0.5cm}
\end{figure}

The possible execution paths for this program lead to two main outputs, where either $\texttt{benefits == true}$ or $\texttt{benefits == false}$.
Thus the best minimiser generated should distinguish the two cases and return a representative value leading to the same output as if the program were to be executed with the ``real'' value assigned to $\texttt{salary}$.
This is the case for the data minimiser described in Fig.~\ref{lst:benefits_simple_minimiser}.
Moreover, we see that any change of value for $\texttt{repr\_salary}$ leads to a change in the $\texttt{salary}$ computed by the main program.
This means this minimiser is indeed a best minimiser.

\begin{figure*}[!t]
    \centering
    \includegraphics[width=0.75\textwidth]{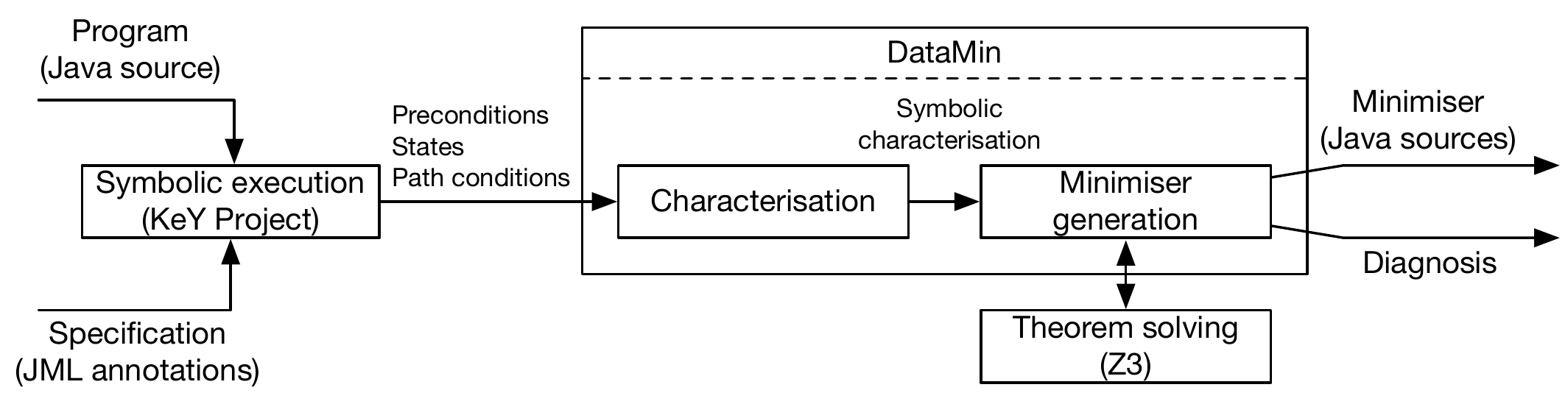}
    \caption{Toolchain of the (semi-)procedure to generate minimisers.}
    \label{fig:toolchain}
    \vspace*{-0.5cm}
\end{figure*}

The procedure to generate data best minimisers follows the toolchain depicted in Fig.~\ref{fig:toolchain}.
Before describing the procedure in detail, we briefly recall concepts related to \emph{symbolic execution}, since we run the program symbolically in order to generate the best minimiser.

\subsection{Symbolic Execution} %
\label{sub:symbolic_execution}

This step corresponds to the two leftmost boxes of the toolchain shown in  Fig.~\ref{fig:toolchain}.
Symbolic execution has been pioneered in~\cite{king1976} as a way to reason about symbolic states reached by a program execution.
A \emph{symbolic execution} of a program is an execution of the program against symbolic inputs.
Expressions are symbolically evaluated and assigned to variables. %
The state of this execution includes the values of the program (modelled as a mapping called {\em store}), and {\em path conditions} (boolean expressions ranging over the symbolic inputs, specifying the conditions to be satisfied for a node to be reached).
Assignments modify the symbolic value of variables, whereas conditions and loops create branches distinguished by their path conditions.
This symbolic execution generates a \emph{symbolic execution tree} where each node is associated to a statement and each arc is associated to a transition of the program.
A state is associated to each node and the root node corresponds to the input of the program.
Each leaf in the tree then corresponds to a completed execution path, and no two leaves have the same path condition.

A key condition for the soundness of symbolic execution is the commutativity of its semantics with respect to the concrete semantics of the language.
This means that executing the program symbolically on given symbolic values and then concretising the output should give the same result as executing the original program over the concretisation of the symbolic inputs. 

The formal definition of this tree makes it possible to reason about the program under analysis.
However, some executions may lead to infinite trees in the presence of loops (when they cannot be simply unfolded).
In this case, in order to get a procedure that terminates and provide conditions for getting a symbolic execution of the program, one should provide \emph{loop invariants} for each concerned loop in the program.
To find the best invariant is a difficult problem and still an active field of research~\cite{flanagan2002}.
Failing to provide the best invariant may result in weaker path conditions failing to capture all of what is theoretically knowable about the states of the symbolic execution.

In our approach the program $P$ is symbolically executed along with the assertions coming from the specification $S$. 
We assume that we have a global \emph{precondition} for the program $P$, given in the specification $S$ of the program, denoted $\mathrm{Pre} \langle P, S \rangle$.
This specification can be expressed explicitly (with annotations in the source code of $P$ for example), or implicitly (by e.g. analyzing the range of the program during the symbolic execution).
In our example, the conditions $0 \leq s$ and $s \leq 100000$ are part of the preconditions as shown in Fig.~\ref{fig:benefits_simple_tree}.
This execution produces a symbolic execution tree equipped with the corresponding path conditions $\mathrm{PC}$ and  $\mathrm{Sto}$ attached to its nodes.

In what follows we define a symbolic characterisation of the program $P$ under specification $S$ capturing the conditions at the termination of the execution.
\begin{definition}    
	[Symbolic characterisation of a program]
	\label{def:symbolic_characterisation}
    We say $\Gamma$ is a \emph{symbolic characterisation} of a program $P$ under specification $S$, written $\Gamma_{\langle P, S \rangle}$, iff $\Gamma$ collects the preconditions, stores, and path conditions to be satisfied for each possible output of $P$:
$\Gamma_{\langle P, S \rangle} \triangleq \mathrm{Pre} \langle P, S \rangle \wedge \left( \bigvee_{l \in \mathrm{Leaves} \langle P, S \rangle} \left( \mathrm{PC}(l) \wedge \mathrm{Sto}(l) \right) \right)$,
where $\mathrm{Leaves}\langle P, S \rangle$ returns the leaves of the symbolic execution tree of $P$ under specification $S$, and $\mathrm{PC}(\cdot)$ and $\mathrm{Sto}(\cdot)$ return the path condition and the state associated to a leaf, respectively.
\end{definition}
\noindent

For the example in Fig.~\ref{lst:benefits_simple_program}, the (simplified) symbolic characterisation is:
$
\Gamma_{\langle P_\mathit{bl}, S_\mathit{bl} \rangle} = 
(0 \leq s \wedge s \leq 100000)\ \wedge\ 
       ((s < 10000 \wedge
                \mathit{salary} = s \wedge
                \mathit{benefits} = \mathit{true})
\lor
       (s \geq 10000 \wedge
                \mathit{salary} = s \wedge 
                \mathit{benefits} = \mathit{false}))
$

The symbolic characterisation induces an equivalence class giving a partition of the state space according to the different possible outputs. This is then used by a solver to get a representative for each equivalence class, which is the next step in our procedure.

\subsection{Static Generation Before Deployment} %
\label{sub:monolithic_data_minimiser_generation}

In this section we show how 
best minimisers %
 can be generated from the symbolic execution tree.
In the following, the monolithic case is seen as a specific case of the distributed case: the former is reduced to a distributed case with only one input.
The construction spans over the two boxes of the toolchain in Fig.~\ref{fig:toolchain} as detailed in the following paragraphs.

The input domain of the semantic function $\textit{dp}$ having $n+1$ inputs is $\prod_{i = 0}^n \mathcal{I}_i$ and the corresponding input variables of $\textit{dp}$ are denoted as $x_i$.
Each domain $\mathcal{I}_i$ is determined implicitly by $\textit{DP}$ or explicitly by its specification $S$.
A local function $\textit{dm}_i$ is generated for each input $x_i \in \mathrm{Input} (\textit{DP})$.
We are still interested in the different possible outputs of $\textit{dp}$ but we cannot directly use its kernel as in the monolithic case since this would require all the minimisers to have access to the data at other data source points.
Instead of this, each data source point assumes that the data to be disclosed by other points can take any value in their domain.
Thus we need to compute, for each input variable $x_i$, the equivalence classes over its domain such that the  value $o$ of the output of the program remains the same for any possible assignment to the other inputs variables $x_j$ for any $j \in (\{0,\dots,n\} \setminus \{i\})$.

 \begin{figure*}[!t]
 \begin{framed}
 \begin{algorithmic}[1]
 \State{$\textit{min} = \{\}$} \Comment initialise minimiser
 \State{$\textrm{decl } \textrm{logical\_variable } i'$} \Comment declare a new logical variable
 \ForAll{$i \in \textit{Inputs}$} \Comment iterate over the inputs
 	\State{$\gamma = \Gamma$} \Comment copy the program symbolic characterisation to $\gamma$
 	\State{$\textit{min}[i] = \{\}$} \Comment initialise minimiser for input $i$
 	\While{$\gamma.\textrm{check()}$} \Comment \strong{call solver} to loop as long as $\gamma$ is satisfiable
 		\State{$\textit{model} = \gamma.\textrm{model()}$} \Comment \strong{call solver} to get a valid assignment for $\gamma$
 		\State{$\textit{formula} = i' == \textit{model}[i] \wedge \Gamma \wedge \Gamma[i'/i]$} \Comment build distributed minimisation formula
 		\State{$\textit{formula} = \exists i'. \forall (\textit{Inputs} \setminus i) \in \textit{Dom}(\textit{Inputs} \setminus i). \exists o. \textit{formula}$} \Comment introduce quantifiers except for $i$
 		\State{$\textit{wp} = \textit{formula}.\textrm{quantif\_elim()}$} \Comment \strong{call solver} to eliminate quantifiers
 		\State{$\textit{min}[i] +\hspace{-1mm}= (\textit{wp}, \textit{model}[i])$} \Comment{add weakest precondition with representative to minimiser}
 		\State{$\gamma = \gamma \wedge \neg \textit{wp}$} \Comment conjunct negation of the weakest precondition to $\gamma$ to limit the loop
 	\EndWhile
 \EndFor
 \end{algorithmic}
 \end{framed}
 \caption{Distributed data minimiser generation (excerpt).}
 \label{alg:data_minimiser_generation_distributed}
     \vspace*{-0.5cm}
 \end{figure*}
 
The algorithm of Fig.~\ref{alg:data_minimiser_generation_distributed} shows how the distribution of the inputs is taken into account. 
Here, $\textit{min}$ (l.1) stands for the minimiser being built (as a map containing the local minimiser $\textit{min}[i]$ (l.5) for each input $i$ (l.3)), $\textit{Inputs}$ (l.3) denotes the inputs of the program, and $\Gamma$ (l.4) denotes the formula of the symbolic characterisation of the program.
The notation $\phi [y/x]$ denotes the formula $\phi$ in which $y$ replaces $x$ and $==$ the logical equality (l.8) while $S \setminus s$ denotes $S \setminus \{s\}$ (the set $S$ from which the element $s$ has been substracted, l.9).

Primitives depending on a theorem prover are called at three locations in the algorithm. The first one, $\textrm{check}()$ (l.6) checks whether or not a logical formula is satisfiable. %
Then, $\textrm{model}()$ (l.7) applied to a formula returns a satisfying valuation. These two primitives are linked as it is only possible to find such a valuation for satisfiable formula.
Finally, $\textrm{quantif\_elim}()$ (l.10) is a procedure to eliminate (both universal and existential) quantifiers in a formula, thus simplifying a formula by removing the corresponding symbols.

After having initialised the data structure $\textit{min}$ holding the minimisers as a map from inputs $i$ to tuples $(\textit{weakest precondition}, \textit{representative})$ (l.1), a new logical variable $i'$ is declared (l.2).
At this point, all the inputs (from $\textit{Inputs}$ (l.3)) and the output $o$ of the program already exist as logical variables (implicitly introduced with $\Gamma$ which is an input of this algorithm).
The rationale to introduce a new logical variable $i'$ is because it is used to control variations on one input during the procedure (l.8-9).
Then, the algorithm iterates over all the inputs (l.3).
The symbolic characterisation $\Gamma$ is assigned to the variable $\gamma$ (l.4) (which will be reinitialised at each iteration of the loop of l.3).
The original $\Gamma$ will be used to build formulas to be solved while the fresher $\gamma$ will be used to control the loop of l.6.
The minimiser $\textit{min}$ is then initialised for the current input $i$ (l.5).
The algorithm should then loop over all equivalence classes for the current input $i$.
This is ensured by (i) representing each equivalence class by its weakest precondition $\textit{wp}$ (l.7-11), (ii) conjuncting the negation of the weakest precondition $\textit{wp}$ found to the symbolic characterisation $\gamma$ (l.11), and (iii) looping as long as there is another equivalence class (ie, weakest precondition) to be covered by checking the satisfiability of $\gamma$ conjuncted with the negation of all previous conditions (l.6).

We now explain in more detail how the weakest preconditions are found (l.7-11).
A satisfying valuation of the characterisation $\gamma$ is requested from the solver and assigned to the variable $\textit{model}$ (l.7).
The valuation of variable $x$ can be called by $\textit{model}[x]$.
A formula (assigned to the variable $\textit{formula}$) is then built in two steps.
First, we conjunct $\Gamma$, $\Gamma[i'/i]$, and the formula $i' == \textit{model}[i]$ (l.8).
This fixes $\Gamma[i'/i]$ to the case where $i'$ is equal to the satisfying valuation found previously (l.7).
Then we introduces quantifiers for all the free logical variables except for $i$ to get the data-based characterisation of data minimisation as given by Proposition~\ref{prop:proper_distribution} (l.9).
The output $o$ is existentially quantified as it can freely vary without any other constraint than those imposed by the $\Gamma$s.
All the inputs except for $i$ and $i'$ are universally quantified because the characterisation states that the preservation of the functional correctness (implied by the value of the output $o$ here) has to be guaranteed for all possible variations of inputs (except for inputs $i$ and $i'$).
Finally, the input $i'$ is existentially quantified only for the purpose of being eliminated later (free variables are existentially quantified for SAT solvers).
Once the formula has been built, the quantifiers are eliminated by calling the suitable procedure from a solver (l.10).
This gives the weakest precondition corresponding to an equivalence class of the inputs (the one corresponding to the value of $\textit{model}[i]$).
This equivalence class $\textit{wp}$ is then added to the minimiser $\textit{min}[i]$ along with its representative $\textit{model}[i]$ (l.11) before being excluded from the control $\gamma$ (l.12) before a new iteration, until exhaustion of equivalence classes (l.6).

\begin{theorem}
    [Soundness]
    The algorithm of Fig.~\ref{alg:data_minimiser_generation_distributed} builds a best distributed minimiser $\textit{dm}$ for program $\textit{dp}$.
\end{theorem}

The soundness of this algorithm relies on Proposition~\ref{prop:best_distr_minimisation}, on the fact that the representative assigned to each equivalence class is fixed, and on the proof of existence in Appendix~A.

\begin{theorem}
    [Termination]
    The algorithm of Fig.~\ref{alg:data_minimiser_generation_distributed} terminates when the number of inputs, the number of equivalence classes on these inputs, and the calls to the external solver terminate.
\end{theorem}

This is proven by showing that all the loops in the algorithm iterate over sets built from inputs and equivalence classes.
The set of inputs is not modified in the loops while the condition of (l.6) strictly decreases in terms of the number of times it is satisfiable.
However, it terminates only if there is a finite number of equivalence classes for the current input.
Since we depend on external procedures, the termination condition also rely on these ones satisfying this condition.
In practice, this is only a semi-procedure.

In both the monolithic and the distributed cases, generating a minimiser may give interesting insights about the process.
It can be used to get metrics, such as the number of equivalence classes finally obtained, to compare with the initial size of the input domain. %
It could also be used 
for quantitative (information-theoretic) analysis.

\subsection{Other Mechanisms}
\label{sub:other_mechanisms}

The mechanism presented in the previous section is the most expensive to deploy.
Indeed, the static approach requires computation of all the (symbolically) possible cases.
Thus, each of the calls to the external solver in the algorithm depicted Fig.~\ref{alg:data_minimiser_generation_distributed} occurs $\sum_{i\in \textit{Inputs}} \#\textit{equivClasses}_i$ times with $\#\textit{equivClasses}_j$ being the number of equivalence classes for the input $i$ (there is in fact one more call for the satisfiability check (l.6)).
This cost is hard to predict in advance.
Indeed, programs which can be highly minimised and end up with few equivalence classes can offer good performance, but quasi-injective programs may have a large number of equivalence classes.

This complexity justifies the need of other mechanisms such as the ones presented in Table~\ref{tab:enforcing_data_minimisation}.
All these other mechanisms allow the reduction of the complexity by limiting the domain of verification.
To effectively use data minimisation, it is thus necessary in general to choose easier mechanisms.
For instance, the static and dynamic mechanisms run online after deployment leads to only $\sum_{i\in \textit{Inputs}} 1 = |\textit{Inputs}|$ calls to the external procedures since the input is known at run-time.
However, this still needs to call semi-decidable external procedures, which may create a large overhead.
If weaker guarantees are acceptable, then it is possible to use test strategies which verify only partial domains.

\section{Implementation and Examples}
\label{sec:implementation_and_examples}

In this section we present \emph{DataMin},
an implementation of the procedure to generate a minimiser.
Two examples are then given to illustrate the notions of data minimality and data minimisation.
The first example contains a loop and lies in the monolithic case while the second example focuses on a distributed setting to illustrate the distributed minimisation.

\subsection{\emph{DataMin} Implementation} %
\label{sub:datamin_implementation}

The (semi-)procedure described in the previous section has been implemented in the Python language as a proof of concept named \emph{DataMin}.%
We rely on the Symbolic Execution Engine of the \emph{KeY} Project~\cite{hentschel2014}.
This symbolic executor is run against a program $P$ written in Java source code for which a minimiser should be generated.

The program can be instrumented with Java Modeling Language (JML) annotations~\cite{leavens1998} to help \emph{KeY} with the exploration of the states space.
They can be used to give a specification $S$ of the domain of inputs under normal behavior for the program according to the design-by-contract terminology~\cite{mitchell2001} (written $\texttt{//@ requires P}$ with $\texttt{P}$ a property assumed to be verified).
The annotations can also be helpful when dealing with loops that cannot be unfolded.
In this case, a loop invariant can be provided (written $\texttt{//@ maintaining P}$, with $\texttt{P}$ the loop invariant) along with a description of a decreasing term to ensure termination (written $\texttt{//@ decreasing t}$, with $\texttt{t}$ the decreasing term).

The symbolic execution results in the definition of stores and path conditions for each leaf of the symbolic execution tree along with the preconditions when applicable.
\emph{DataMin} generates the symbolic characterisation $\Gamma_{\langle P, S \rangle}$ and builds the partitioning $k_{\langle P, S \rangle}$ and the sectioning $r_{\langle P, S \rangle}$ functions.
The theorem prover \emph{Z3}~\cite{demoura2008} is called through its API to solve constrains as needed during these steps.
This proof-of-concept supports only a limited range of data structures but could be extended thanks to the many theories on which \emph{Z3} relies.%

Finally, \emph{DataMin} generates the minimiser as a set of Java source code files to be run by the data source points before disclosing the data.
These files are directly compilable and ready to be exported as Java libraries to ease their use.
This whole process currently runs in reasonable time for these examples (less than a second, given that Python is run in its interpreted mode). 
For the first example, the solver used (\emph{Z3}) is not able to totally perform quantification elimination.
We thus checked the results hold by manually verifying them using the \emph{Redlog} system~\cite{dolzmann1997}.
This is a limitation of the external tools and not of the procedure proposed in this paper.
\subsection{Loyalty Status} %
\label{sub:loyalty_status}

The first example proposes a situation where an airport facility must deliver different services to different customers depending on their status.
This status level is determined depending on the number of flights made by the customer in the previous year with its favorite company, \emph{PrivaFly}.
This number of flights is disclosed by the company to the airport.

However, \emph{PrivaFly} wants to adopt the best practices in personal data protection and requires that only the needed data should be disclosed.
The airport services have their own policy to compute the status level, as shown in Fig.~\ref{lst:loyaltyapp_program}.
The method $\texttt{compStatusLevel}$ takes a number of $\texttt{flights}$ and returns a $\texttt{status}$ level.
If the number of flights is lower than $9$, the status is $0$. From $10$ to $19$ flights, the status level is equal to the number of flights from which $10$ is substracted.
For a number of flights ranging from $20$ to $29$, the computation involves a loop and will be detailed later.
Finally, over $30$ flights, the status level is capped to $500$.
The first JML annotation (l.2) specifies a range for the input domain.
Then, the following annotations from (l.9-13) will be described later.

\begin{figure}[t!]
    \centering
\begin{lstlisting}[style=my_style,columns=flexible]
public class LoyaltyApp {
    //@ requires (0 <= flights) && (flights <= 100);
    public int compStatusLevel(int flights) {
        int status;
        if (flights < 10) { status = 0; }
        else if (flights < 20) { status = (flights - 10); }
        else if (flights < 30) {
            status = 0; int i = 0;
			//@ maintaining (0 <= i)
            //@          && (i <= flights - 19)
            //@          && (status == i * flights);
			//@ decreasing (flights - 19 - i);
			//@ assignable \strictly_nothing;
            while (i <= flights - 20) {
                status = (status + flights);
                i = i + 1; }
            if (status > 150) { status = 150; } }
        else { status = 500; }
        return status; } }
\end{lstlisting}
    \caption{Program $P_\textit{ls}$ to compute a loyalty status.}
    \label{lst:loyaltyapp_program}
    \vspace*{-0.5cm}
\end{figure}

Intuitively, there is no need to give a precise value for a number of flights between $0$ and $9$ and over $30$.
On the other hand, the exact number should be disclosed between $10$ and $19$.
The program $M_{\textit{ls}}$ from Fig.~\ref{lst:loyaltyapp_minimiser} is the result of running \emph{DataMin} on $P_\textit{ls}$ (see Fig.~\ref{lst:loyaltyapp_program}).
In this program, the complex (weakest pre)conditions for each equivalence class returned by \emph{DataMin} have been replaced by placeholders $\texttt{wp}i$ (with $i$ a number).
Each of these conditions depends only on the input of the minimiser (here, on the variable $\texttt{flights}$) in order to be executable.

\begin{figure}[t!]
    \centering
\begin{lstlisting}[style=my_style,columns=flexible]
public class LoyaltyAppMin_flights {
    public int minimise_compStatusLevel(int flights) {
        int repr_flights;
        if (wp1) { repr_flights = 0; }
        else if (wp2) { repr_flights = 11; }|\Suppressnumber|
        ...|\Reactivatenumber{18}|
        else if (wp15) { repr_flights = 24; }|\Reactivatenumber{19}|
        else if (wp16) { repr_flights = 25; }|\Reactivatenumber{20}|
        else { repr_flights = 30; }|\Reactivatenumber{21}|
        return repr_flights; } }|\Reactivatenumber{1}|
\end{lstlisting}
    \caption{Minimiser $M_{\textit{ls}}$ for $P_\textit{ls}$ (see Fig.~\ref{lst:loyaltyapp_program}).}
    \label{lst:loyaltyapp_minimiser}
    \vspace*{-0.5cm}
\end{figure}

We can see that the intuition above is partially confirmed by the minimiser generated.
The representatives chosen by the minimiser for $\texttt{repr\_flights}$ are the minimum of the corresponding equivalence classes.
As a consequence it has been fixed to $0$ when the actual number of flights ranges from $0$ to $10$ (l.4 of the minimiser).
We observe $10$ is included (case not captured by the following conditionals of the minimiser) despite the fact it is not captured by the first conditional of the program (l.5) but by the second one (l.6).
However, since the status level is also equal to $0$ in this case, it belongs to the same equivalence class as the number of flights captured by the first conditional.
For a number of flights between $11$ and $19$, the output is different each time  (number of flights minus $10$).
Thus the representative is the number of flights itself (l.5-13 of the minimiser).
The case of a number of flights greater than $30$ is trivially captured by a representative fixed to $30$ as well (l.20).

Between $20$ and $29$ flights, the status is computed thanks to a loop (l.14-16 of the program).
The code makes the status level computed after this point to be equal to the number of flights multiplied by the same number minus $19$.
The symbolic executor of \emph{KeY} is not able to solve this loop itself and an invariant ($\texttt{maintaining}$) has been provided (l.9-11).
These relations described by the loop invariant are true both at the beginning and at the end of each iteration.
The following annotation ($\texttt{decreasing}$, l.12) helps the executor with the termination of the loop by providing a decreasing term. 
Finally, the last annotation ($\texttt{assignable}$, l.13) also helps the symbolic executor to handle the loop (by declaring here that no side effect occurs).
These annotations give to the symbolic executor the best invariant.
As a consequence, the data minimisation process is able to propose a minimiser.
The representatives for a number of flights from $20$ to $24$ correspond to the number of flights itself (l.14-18 of the minimiser).
This is not the case starting from $25$ until $29$ (l.19 of the minimiser) because the status is capped at $150$ when the number of flights is less than $30$ and these five possibilities constitute their own equivalence class.

\begin{table*}[t!]
\centering
\begin{threeparttable}
\caption{Equivalence classes for $P_\textit{ls}$ (see Fig.~\ref{lst:loyaltyapp_program}).}
\label{tab:loyalty}
\begin{tabular}{|c|c|c|c|c|c|c|c|}
\hline
\multirow{2}{*}{\textit{invariant}} & \textit{variables} $\backslash$ \textit{values} & \multicolumn{1}{c|}{$[0, 9]$}         & \multicolumn{1}{c|}{$10$} & \multicolumn{1}{c|}{$[11, 19]$} & \multicolumn{1}{c|}{$[20, 24]$} & \multicolumn{1}{c|}{$[25, 29]$} & \multicolumn{1}{c|}{$[30, 100]$} \\ \cline{2-8} 
                                    & $\texttt{flights}$       & 1\textsuperscript{st} $\texttt{if}$ & \multicolumn{2}{c|}{2\textsuperscript{nd} $\texttt{if}$}  & \multicolumn{2}{c|}{3\textsuperscript{rd} $\texttt{if}$}        & final $\texttt{else}$            \\ \hline
best                                & $\texttt{repr\_flights}$ & \multicolumn{2}{c|}{1 class}                                      & 9 classes                       & 5 classes                       & 1 class                         & 1 class                          \\ \hline
weak                                & $\texttt{repr\_flights}$ & \multicolumn{2}{c|}{1 class}                                      & 9 classes                       & 5 classes                       & 5 classes                       & 1 class                          \\ \hline
\end{tabular}
\begin{tablenotes}
    \item Note: Correspondance between $\texttt{flights}$ and their representatives, mismatch between syntactic structure (\texttt{if} conditionals) and semantic classes, and effects of weak invariants on the minimisation.
\end{tablenotes}
\end{threeparttable}
\end{table*}

If the invariant provided for the loop had not been the best invariant, the equivalence class composed of a number of flights between $25$ and $29$ may not have been detected.
The representatives for these cases would have been the actual number, leading to more disclosure than needed.

The observations made during this example are summed up in Table~\ref{tab:loyalty}.
The table shows the mismatch between the syntactic structure of a program (modelled in this table by the correspondence between the number of flights and the $\texttt{if}$ conditionals which are triggered) and the semantic equivalence classes.
The effect of strengthening an invariant is clear, and it shows the impact it may have on the resulting pre-processor.

We have shown here an example of a monolithic minimiser.
The following example addresses the distributed case, when data comes from different sources.

\subsection{Credit Score} %
\label{sub:credit_score}

This second example takes place from the point of view of \emph{PrivaLoan}, a company specialised in making loans to individual customers. 
To decide whether or not loan is granted, \emph{PrivaLoan} computes a credit score.
This score depends on two elements: the credit history and the tax category of the individual.
These two inputs are disclosed by two different parties: a central bank for the credit history and the tax office for the tax category.

One of the competitive advantage of \emph{PrivaLoan} on the credit market is that it protects the privacy of the individuals asking for credit.
This is achieved by applying the limited collection principle to personal data requested to get an offer.
The rules to decide on the credit risk of an individual are implemented in the $\texttt{compCreditScore}$ method of the program from Fig.~\ref{lst:creditapp_program}.
This method takes two arguments: $\texttt{incidents}$, which ranges from $0$ to $3$ as a grade for the credit-history of an individual, and $\texttt{tax}$, which ranges from $1$ to $3$ to indicate the tax category to which the individual belongs.
For an incident level of $2$ or $3$, the credit score is $0$, which means \emph{PrivaLoan} does not want to make an offer.
If the incident level is $0$ or $1$, the credit score is $1$, except when the incident level is $0$ and the tax category is $3$, which leads to a credit score of $2$.

\begin{figure}[t!]
    \centering
\begin{lstlisting}[style=my_style,columns=flexible]
public class CreditApp {|\Reactivatenumber{2}|
    //@ requires (0 <= incidents) && (incidents <= 3)|\Reactivatenumber{3}|
    //@       && (1 <= tax)  && (tax <= 3);|\Reactivatenumber{4}|
    public int compCreditScore(int incidents, int tax) {|\Reactivatenumber{5}|
        int creditScore; |\Reactivatenumber{6}|
        if (incidents == 0) { |\Reactivatenumber{7}|
            if (tax == 3) { creditScore = 2; } |\Reactivatenumber{8}|
            else { creditScore = 1; } } |\Reactivatenumber{9}|
        else if (incidents == 1) { creditScore = 1; } |\Reactivatenumber{10}|
        else { creditScore = 0; } |\Reactivatenumber{11}|
        return creditScore; } } |\Reactivatenumber{1}|
\end{lstlisting}
    \caption{Program $P_\textit{cs}$ to compute a credit score.}
    \label{lst:creditapp_program}
    \vspace*{-0.5cm}
\end{figure}

As  can be seen in Table~\ref{tab:credit}, which maps the value of $\texttt{creditScore}$ as function of $\texttt{incidents}$ and $\texttt{tax}$, some combinations of inputs return the same score.
Intuitively, there is no need for the central bank to make a distinction between incidents of grade $2$ and $3$ since the output is the same.
Similarly, the tax office should not make a distinction between tax categories of $1$ and $2$.
Indeed, for all the possible values of incidents, the credit score does not change between a tax category of $1$ and one of $2$ (though this score does not remain the same for all values of incidents).

\begin{table}[t!]
\centering
\begin{threeparttable}
\caption{Semantics of $P_\textit{cs}$ (see Fig.~\ref{lst:creditapp_program}).}
\label{tab:credit}
\begin{tabular}{|l|c|cl|c|c|}
\hline
\multicolumn{2}{|l|}{\multirow{1}{*}{\emph{variables}}} & \multicolumn{4}{c|}{$\texttt{incidents}$}                                                       \\ \cline{2-6} 
\multicolumn{1}{|c|}{}  &     \emph{values}          & \multicolumn{1}{c|}{~~~$0$~~~~} & \multicolumn{1}{c|}{~~~$1$~~~~} & ~~~$2$~~~~                 & ~~~$3$~~~~                 \\ \hline
\multirow{3}{*}{~$\texttt{tax}$~~}       & $1$      & \multicolumn{2}{c|}{\multirow{2}{*}{$1$}}            & \multicolumn{2}{c|}{\multirow{3}{*}{$0$}} \\ \cline{2-2}
                                      & $2$      & \multicolumn{2}{c|}{}                                & \multicolumn{2}{c|}{}                     \\ \cline{2-3}
                                      & $3$      & \multicolumn{1}{c|}{$2$} &                          & \multicolumn{2}{c|}{}                     \\ \hline 
\end{tabular}
\begin{tablenotes}
    \item Note: There are three equivalence classes corresponding to the output values $0$, $1$, and $2$.
\end{tablenotes}
\end{threeparttable}
\end{table}

The minimisers $M_{\textit{cs},0}$ and $M_{\textit{cs},1}$ in Fig.~\ref{lst:creditapp_minimiser_incidents} and~\ref{lst:creditapp_minimiser_tax} are output by \emph{DataMin} based on $P_\mathit{cs}$ from Fig.~\ref{lst:creditapp_program}.
Both minimisers reflect the behavior described previously.
As previously, we replaces the (weakest pre-)conditions by placeholders.
Indeed, the representative assigned for the equivalence class $\{2, 3\}$ corresponds to the minimum value of this set (l.5 of the minimiser for the variable $\texttt{incidents}$, see Fig.~\ref{lst:creditapp_minimiser_incidents}).
In all these cases, the value of the output of the program $P_\mathit{cs}$ execution remains the same whichever the representative value chosen.

\begin{figure}[t!] 
    \centering
\begin{lstlisting}[style=my_style,columns=flexible]
public class CreditAppMin_incidents { |\Reactivatenumber{2}|
    public int minimise_compCreditScore(int incidents) { |\Reactivatenumber{3}|
        int repr_incidents; |\Reactivatenumber{4}|
        if (wp1) { repr_incidents = 1; } |\Reactivatenumber{5}|
        else if (wp2) { repr_incidents = 2; } |\Reactivatenumber{6}|
        else { repr_incidents = 0; } |\Reactivatenumber{7}|
        return repr_incidents; } } |\Reactivatenumber{1}|
\end{lstlisting}
    \caption{First input ($\textit{incidents}$) minimiser $M_{\textit{cs},0}$ for $P_\textit{cs}$ (see Fig.~\ref{lst:creditapp_program}).}
    \label{lst:creditapp_minimiser_incidents}
\end{figure}

Similarly for the minimiser for the variable $\texttt{tax}$ (see Fig.~\ref{lst:creditapp_minimiser_tax}), the representative computed for the equivalence class $\{ 1, 2 \}$ (l.4) can have any of these values without changing the value of the final output ($1$ has been fixed here).

\begin{figure}[t!]
    \centering
\begin{lstlisting}[style=my_style,columns=flexible]
public class CreditAppMin_tax { |\Reactivatenumber{2}|
    public int minimise_compCreditScore(int tax) { |\Reactivatenumber{3}|
        int repr_tax; |\Reactivatenumber{4}|
        if (wp1) { repr_tax = 1; }|\Reactivatenumber{5}|
        else { repr_tax = 3; }|\Reactivatenumber{6}|
        return repr_tax; } } |\Reactivatenumber{1}|
\end{lstlisting}
    \caption{Second input ($\textit{tax}$) minimiser $M_{\textit{cs},1}$ for $P_\textit{cs}$ (see Fig.~\ref{lst:creditapp_program}).}
    \label{lst:creditapp_minimiser_tax}
    \vspace*{-0.5cm}
\end{figure}

In all the other cases, the representative must have the same value as the input variable (l.4 and l.6 of $M_{\textit{cs},0}$ and l.5 of $M_{\textit{cs},1}$).
Indeed, without knowing a priori the value that will be disclosed by the other data source point, each has to assume it could be any value.

\section{Related Work} %
\label{sec:related_works}

Formal and rigorous approaches to privacy have long been advocated~\cite{tschantz2009}. Many different dimensions in this field have been explored but the data minimisation principle seems not to have been precisely defined, as witnessed by~\cite{gurses2015}, where it is stated that ``a whole family of design principles are lumped under the term `data minimisation'''. 
A related work is the notion of \emph{minimal exposure}~\cite{anciaux2012}. This consists in performing a preprocessing of information on the client's side to give only the data needed to benefit from a service. For instance, if the service to which an individual requests can be modeled as $a \vee b$, then if $a$ is true it will be the only data disclosed: they reduce the \emph{number} of inputs provided, but do not reason on the domain of the inputs.

The semantics of minimality and data minimisation is closely related to information flow, and we have used several ideas from that area \cite{cohen1977,Landauer:Redmond:CSFW93,sabelfeld2001,giacobazzi2004}. 
A semantic notion of dependency has been introduced in \cite{mastroeni2008}, where it is used to characterise the subset of variables in an expression containing all and only the variables that are semantically relevant for the evaluation of the expression. Our notion is related to this (viewing the variables of an expression as the input, and the value of the expression at the output), but much more fine-grained.

A notable difference in the formalisation of minimality compared to usual definitions of information flow is that minimality is a necessary and sufficient condition on information flow, whereas most security formulations of information flow are sufficient conditions only (e.g. noninterference: public inputs are sufficient to compute public outputs). An exception here is Chong's \emph{necessary information release} \cite{chong2010} which provides both upper and lower bounds on information flows. For this reason many static analysis techniques that have been used for security properties (e.g. type systems) are not easy to use for minimality -- an over-approximation to the information flows is not sound for minimality.

The necessary and sufficient conditions embodied in minimisers appeasr to be closely related to the notion of \emph{completeness} in abstract interpretation \cite{Giacobazzi:Ranzato:Completeness}, where a minimiser $m$ plays the role of a \emph{complete abstraction}. 
It is to be noted that some work in the domain of quantitative information flow aiming at automated discovery of leaks also rely on analysis of equivalence classes \cite{Lowe:CSFW02,backes2009}. Compared to our proof-of-concept tool, the approach of \cite{backes2009} contains several ideas that could be beneficial to adopt in our tool.

The idea of equivalence partitioning by using symbolic execution has been used previously in the context of test case generation in \cite{richardson1981}. Testing is also one of the uses of the symbolic execution as performed by the \emph{KeY} theorem prover \cite{KeY2005}. We use the \emph{KeY} tool in our work to symbolically execute the given program, also exploiting the possibility of adding preconditions and invariants.

\section{Conclusion} %
\label{sec:conclusion}

We have provided a formal definition of {\it data minimisation} in terms of strong dependency and derived concepts. We have also formalised the notion of a {\it data minimiser} as a pre-processor that may be composed with the data processor to filter the inputs so the latter receives representative inputs for different equivalence classes yielding the same output. We have considered the cases when the input comes from a unique source (monolithic collection), and when they come from different sources (distributed collection). We have also shown properties of both kinds of data minimisers depending on when they are used (before deployment, online, or offline).

Besides the theoretical results, we have also introduced a procedure, and a proof-of-concept implementation, to obtain data minimisers for a given program. The procedure is a chain consisting of a symbolic execution of the program and the use of a SAT solver to obtain a partition of the input space depending on the different possible outputs. The procedure generates representatives for each class which are outputs of the generated data minimiser to be fed as inputs into the data processor.

Given the nature of the problem our approach is semantics-based, so finding a distributed minimiser is undecidable in general. That said, as data has different sensitivity depending on the application, and the desirable privacy level, we could restrict the analysis to specific modules, making it tractable in some cases. The approach might also have practical utility if we aim at using it in an interactive way, semi-automatically. 
One crucial issue in tackling realistic examples will be finding compositional methods. One possible compositional proof rule is the following: 
If $m$ is a minimiser for $p$, $n$ is a (distributed) minimiser for $q$ and
$\range(q) \subseteq \range(m)$ then 
$n$ is a (distributed) minimiser for $p \circ q$.

Note the special case when $p$ is minimal (and hence $m$ is the identity) for which 
the condition $\range(q) \subseteq \range(m)$ becomes vacuous. 
It remains to be explored whether this has practical utility. 

Our paper opens the way to other interesting research directions. On the theoretical side it would be worth exploring the connections to completeness in abstract interpretation as mentioned in the related work section. On the practical side we are considering the possibility of defining a type system to determine the cardinality of the needed inputs for computing a given output, and also the integration of a quantitative approach based on worth-based information theory to deal with the sensitivity of data.

\newpage

\bibliographystyle{splncs03}
\bibliography{extracted}

\section*{Appendix A: Proof of Theorem~\ref{th:distributed_minimiser_existence}}
\label{sec:appendix_proof_of_theorem_ref_th_distributed_minimiser_existence}

Let $\mathrm{ER}(D)$ denote the equivalence relations over set $D$.

We view equivalence relations as sets of pairs, writing either $a R b$ or $(a, b) \in R$ when relation $R$ relates $a$ to $b$. Note that $\mathrm{ER}(D)$ forms a complete lattice $\langle \mathrm{ER}(D), \bot, \top, \leq, \vee, \wedge \rangle$ where $\bot = \left\{ (a, a) \mid a \in D \right\}$, $\top = D \times D$, $\leq$ is subset inclusion $\subseteq$, $\vee$ is transitive closure of set union $\cup^*$, and $\wedge$ is set intersection $\cap$.

\begin{definition}
    Given $R \in \mathrm{ER}(D)$ and $S \in \mathrm{ER}(E)$ let $R \boxtimes S \in \mathrm{ER}(D \times E)$ denote the relation $(d, e) R \boxtimes S (d^{\prime}, e^{\prime}) \Leftrightarrow d R d^{\prime} \wedge e S e^{\prime}$.
\end{definition}

The lattice of equivalence relations $\mathrm{ER}(D)$ forms a complete lattice (standard) as described above.

\begin{proposition}
    \label{prop:complete_sub-lattice}
    $\left\{ \underset{{i \in \{0, \dots, n\}}}{\boxtimes} R_i \mid R_i \in \mathrm{ER}(D_i) \right\}$ is a complete sub-lattice of\\ $\mathrm{ER} \left( \underset{i \in \{0, \dots, n\}}{\prod} D_i \right)$.
\end{proposition}

\begin{proof}
    Note that in $\mathrm{ER}(D)$, $\bot$ is the identity on $D$, written $\mathrm{Id}_D$ and $\top$ is the universal relation $\mathrm{All}_D$.
    
    It is easy to see that $\underset{{i \in \{0, \dots, n\}}}{\boxtimes} \mathrm{Id}_{D_i} = \mathrm{Id}_{\left( \underset{{i \in \{0, \dots, n\}}}{\boxtimes} D_i \right)}$ and $\underset{{i \in \{0, \dots, n\}}}{\boxtimes} \mathrm{All}_{D_i} = \mathrm{All}_{\left( \underset{{i \in \{0, \dots, n\}}}{\boxtimes} D_i \right)}$.
    
    Similarly, it is straightforward to show that \\ $\underset{{i \in \{0, \dots, n\}}}{\boxtimes} R_i \cup \underset{{i \in \{0, \dots, n\}}}{\boxtimes} S_i = \underset{{i \in \{0, \dots, n\}}}{\boxtimes} (R_i \cup S_i)$ and similarly for $\cap$ (by observing that union, intersection, and transitive closure all act pointwise on $t$-uples).
\end{proof}

{\bf PROOF OF THE THEOREM}

We need to prove that {\it for any $\textit{dp} \in \left( \underset{i \in \{0, \dots, n\}}{\prod} I_i \right) \rightarrow O$, $\exists \textit{dm} = \underset{i \in \{0, \dots, n\}}{\otimes} m_i$ such that $\textit{dm}$ is a best distributed minimiser for $\textit{dp}$.}

\begin{proof}
    We first build a candidate $\textit{dm}$, and by assuming that it is \emph{not} a best distributed minimiser, derive a contradiction.
    
    Let $E = \left\{ \underset{{i \in \{0, \dots, n\}}}{\boxtimes} R_i \mid R_i \in \mathrm{ER}(I_i) \right\}$. Note that $E$ is a complete sub-lattice of $\mathrm{ER} \left( \underset{i \in \{0, \dots, n\}}{\prod} I_i \right)$ (by Proposition~\ref{prop:complete_sub-lattice}).
    
    Define $K = \cup^* \{ R \in E \mid R \subseteq \mathrm{ker} \, \textit{dp} \}$.
    
    By the sub-lattice property of $E$ above, $K \in E$, ie. $K = \underset{i \in \{0, \dots, n\}}{\prod} K_i$ for some $K_i \in \mathrm{ER}(I_i)$ (for $i \in \{0, \dots, n\}$) and $K \subseteq \mathrm{ker} \, \textit{dp}$ (since $\mathrm{ker} \, \textit{dp}$ is an upper bound, it is also an upper bound for the least upper bound).
    
    Now we define a candidate $\textit{dm} = \underset{i \in \{0, \dots, n\}}{\otimes} m_i$. To do so, for each equivalence class $C$ of each $K_i$ (for $i \in \{0, \dots, n\}$), choose a representative element denoted $\mathrm{repr}(C)$.
    
    Now define $m_i(v) \triangleq \mathrm{repr}\left( [v]_{K_i} \right)$, ie. $m_i$ maps a value into the representative chosen for that value's equivalence class under $K_i$.
    
    \strong{First} we show that $\textit{dm}$ is a minimiser for $P$.
    
    \begin{align}
        \textit{dp}(\textit{dm}(v)) &= \textit{dp}\left(m_0(v_0), \dots, m_n(v_n)\right) \nonumber\\
        &= \textit{dp}\left(\mathrm{repr}\left( [v_0]_{K_0} \right), \dots, \mathrm{repr}\left( [v_n]_{K_n} \right) \right)\nonumber
    \end{align}
    
    Note that $\left( \mathrm{repr}\left( [v_i]_{K_i} \right), v_i \right) \in K_i$ (for $i \in \{0, \dots, n\}$) then \\ $\left(\mathrm{repr}\left( [v_0]_{K_0} \right), \dots, \mathrm{repr}\left( [v_n]_{K_n} \right) \right)$ is related to $v$ by $K$ and hence by $\mathrm{ker} \, \textit{dp}$.
    
    Thus $\textit{dp}\left(\mathrm{repr}\left( [v_0]_{K_0} \right), \dots, \mathrm{repr}\left( [v_n]_{K_n} \right) \right) = \textit{dp}(v)$.
    
    It remains to show that $\textit{dm}$ is idempotent. This follows easily from the fact that the $m_i$ are idempotent since $\mathrm{repr}(C) \in C$ for any equivalence class of $m_i$.
    
    \strong{Finally}, assume towards a contradiction that $\textit{dm}$ thus defined is \emph{not} a best distributed minimiser for $\textit{dp}$.
    
    Let $J_i = \mathrm{range} \left( m_i \right)$. Then $\mathrm{range} \left( m_i \right) = \underset{i \in \{0, \dots, n\}}{\prod} J_i$. Then by the assumption $\textit{dp}$ is \emph{not} best distributed-minimised for $\underset{i \in \{0, \dots, n\}}{\prod} J_i$, $\exists a \in \{0, \dots, n\}$ and distinct values $u, v \in J_a$ such that $\forall \overline{u}, \overline{v}$ which differ only at position $a$, where $\overline{u}_a = u, \overline{v}_a = v$ we have $\textit{dp} \left( \overline{u} \right) = \textit{dp}\left( \overline{v} \right)$.
    
    Define $R_a = \{ (u, v), (v, u) \} \cup^* \mathrm{Id}_{J_a}$. Define $R = \underset{i \in \{0, \dots, n\}}{\boxtimes} R_i$ where $R_i = \mathrm{Id}_{J_i}$ for $i \neq a$.
    
    Thus $\left( \overline{u}, \overline{v} \right) \in R \Rightarrow \textit{dp} \left( \overline{u} \right) = \textit{dp}\left( \overline{v} \right)$ and hence $R \subseteq \mathrm{ker} \, \textit{dp} \mid_{\mathrm{range}\left( \textit{dm} \right)}$, and hence (*) $R \subseteq K$.
    
    By construction of $m_a$, since $u \neq v$, then $u$ and $v$ must not be related by $K_a$ which contradicts (*).
\end{proof}

\end{document}